\documentclass[journal]{IEEEtran}

\usepackage{rotating}
\usepackage{chngcntr}
\usepackage{apptools}
\AtAppendix{\counterwithin{thm}{section}}
\usepackage{graphicx}
\usepackage{graphics}
\usepackage{amssymb}
\usepackage{amsmath}
\usepackage{mathtools}
\usepackage{color}
\usepackage{xspace}
\usepackage{algpseudocode}
\usepackage{bbm}
\usepackage{comment}
\usepackage{algorithmicx}
\usepackage{subfig}
\usepackage{psfrag}

\usepackage{tikz}
\usetikzlibrary{shapes,arrows}
\usetikzlibrary{positioning}

\usepackage{rotating}
\usepackage{chngcntr}
\usepackage{apptools}
\usepackage{listings}
\AtAppendix{\counterwithin{thm}{section}}
\usepackage{graphicx}
\usepackage{graphics}
\usepackage{amssymb}
\usepackage{amsmath}
\usepackage{color}
\usepackage{xspace}
\usepackage{algpseudocode}
\usepackage{bbm}
\usepackage{comment}
\usepackage{algorithmicx}
\usepackage{subfig}
\usepackage{psfrag}
\usepackage{tikz}
\usetikzlibrary{shapes,arrows}
\usetikzlibrary{positioning}
\let\labelindent\relax
\usepackage{enumitem}

\DeclarePairedDelimiter{\ceil}{\lceil}{\rceil}
\tikzstyle{block}=[draw opacity=0.7,line width=1.4cm]
\DeclareMathAlphabet{\mathpzc}{OT1}{pzc}{m}{it}
\definecolor{CranJ}{cmyk}{0,0.69,0.54,0.04} 
\definecolor{PinkJ}{cmyk}{0,0.71,0.43,0.12} 
\definecolor{Cran}{cmyk}{0,0.73,0.41,0.29} 
\definecolor{VRed}{cmyk}{0,0.75,0.25,0.2} 
\definecolor{ORed}{cmyk}{0,0.75,0.75,0} 
\definecolor{CBlue}{cmyk}{1,0.25,0,0} 

\title{\LARGE \bf 
On Robustness Analysis of a Dynamic Average Consensus Algorithm to Communication Delay}

\author{Hossein Moradian and Solmaz S. Kia %
  \thanks{The authors are with the Department of Mechanical and Aerospace Engineering, University of California Irvine, Irvine, CA 92697,  
    {\tt\small \{hmoradia,solmaz\}@uci.edu}. This work is supported by NSF award ECCS-1653838.}%
}

\newcommand{\ee}{\operatorname{e}}
\newcommand{\ii}{\text{i}}
\newcommand{\VV}{\mathcal{V}}
\newcommand{\EE}{\mathcal{E}}
\newcommand{\GG}{\mathcal{G}}

\newcommand{\lL}{\vectsf{L}}
\newcommand{\rR}{\vect{{R}}}

\newcommand{\PPi}{\vect{\Pi}}

\newcommand{\real}{{\mathbb{R}}}
\newcommand{\reals}{{\mathbb{R}}}

\newcommand{\complex}{{\mathbb{C}}} 
\newcommand{\realpositive}{{\mathbb{R}}_{>0}}
\newcommand{\realnonnegative}{{\mathbb{R}}_{\ge 0}}

\newcommand{\eps}{\epsilon}

\newcommand{\Hlambda}{\hat{\lambda}}
\newcommand{\argmin}{\operatorname{argmin}}

\newcommand{\dout}{\mathsf{d}_{\operatorname{out}}}
\newcommand{\dM}{\mathsf{d}^{\max}}
\newcommand{\re}[1]{\operatorname{Re}(#1)}
\newcommand{\im}[1]{\operatorname{Im}(#1)}

\newcommand{\kappaa}{\mathpzc{k}}

\newcommand{\until}[1]{\in\{1,\dots,#1\}}

\renewcommand{\baselinestretch}{.985}

\newcommand{\vect}[1]{\boldsymbol{\mathbf{#1}}}
\newcommand{\vectsf}[1]{\vect{\mathsf{#1}}}
\newcommand{\Bvect}[1]{\bar{\boldsymbol{\mathbf{#1}}}}

\newcommand{\dvect}[1]{\dot{\vect{#1}}}
\newcommand{\dvectsf}[1]{\dot{\vectsf{#1}}}

\newcommand{\Sym}[1]{\operatorname{Sym}(#1)}
\newcommand{\Diag}[1]{\operatorname{Diag}(#1)}

\newcommand{\avrg}[1]{\frac{1}{N}\sum\nolimits_{j=1}^N#1^j}

\newtheorem{thm}{Theorem}[section]

\newtheorem{lem}{Lemma}[section]
\newtheorem{defn}{Definition}

\newcommand{\oprocendsymbol}{\hbox{$\bullet$}}
\newcommand{\oprocend}{\relax\ifmmode\else\unskip\hfill\fi\oprocendsymbol}


\ifCLASSINFOpdf
\else
\fi

\begin{document}

\maketitle
\begin{abstract}
This paper studies the robustness of a dynamic average consensus algorithm to communication delay over strongly connected and weight-balanced (SCWB) digraphs. Under delay-free communication, the algorithm of interest achieves a practical asymptotic tracking of the dynamic average of the time-varying agents' reference signals. For this algorithm, in both its continuous-time and discrete-time implementations, we characterize the admissible communication delay range and study the effect of the delay on the rate of convergence and the tracking error bound. Our study also includes establishing a relationship between the admissible delay bound and the maximum degree of the SCWB digraphs. We also show that for delays in the admissible bound, for static signals the algorithms achieve perfect tracking. Moreover, when the interaction topology is a connected  undirected graph, we show that the discrete-time implementation is guaranteed to tolerate at least one step delay. Simulations demonstrate our results. 
\end{abstract}


\section{Introduction}
In a network of agents each endowed with a dynamic reference input signal, the dynamic average consensus problem consists of designing a distributed algorithm that allows each agent to track the dynamic average of the reference inputs (a global task) across the network. The solution to this problem is of interest in numerous applications such as multi-robot coordination~\cite{PY-RAF-KML:08},
sensor fusion~\cite{ROS-JSS:05,ROS:07,WR-UMA:17}, distributed optimal resource allocation~\cite{AC-JC:14-auto,SSK:17}, distributed estimation~\cite{SM:07},
and distributed tracking~\cite{PY-RAF-KML:07}. Motivated by the fact that delays are inevitable in real systems, our aim here is to study the robustness of a dynamic average consensus algorithm to fixed communication delays. The methods we develop can be applied to other dynamic consensus algorithms, as well. 

Distributed solutions to the dynamic and the static average consensus problems have attracted increasing attention in the last decade. 
Static average consensus, in which the reference signal at each agent is a constant static value, has been studied extensively in the literature (see e.g.,~\cite{ROS-RMM:04,WR-RWB:05,LX-SB:04}). Many aspects of the static average consensus problem including analyzing the convergence of the proposed algorithms in the presence of communication delays are examined in the literature (see e.g.,~\cite{ROS-RMM:04,PAB-GFT:08,AS-DVD-KHJ:08,CNH-TC:14,YG-NM:17}). Also, the static average consensus problem for the multi-agent systems with second-order dynamics in the presence of communication delay has been studied in~\cite{WY-ALB-XW:08,JC-MC-ZG-RL-DZ:14,ZM-WR-YC-ZY:11}. The dynamic average consensus problem has been studied in the literature, as well. The solutions for this problem normally guarantee convergence to some neighborhood of the network's dynamic average of the reference signals (see. e.g.,~\cite{ROS-JSS:05,dps-ros-rmm:05b,RAF-PY-KML:06,SSK-JC-SM:15-ijrnc,FC-WR-WL-GC:15,SSK-JC-SM:15-auto2} for the continuous-time algorithms and~\cite{MZ-SM:08a,EM-CS-JIM-SM:14,SSK-JC-SM:15-ijrnc} for the discrete-time algorithms). Zero error tracking has been achieved under restrictive assumptions on the type of the reference signals~\cite{HB-RAF-KML:10} or via non-smooth algorithms, which assume an upper bound on  the derivative of the agents' reference signals is known~\cite{FC-YC-WR:12}. Some of these references address important practical considerations such as the dynamic average consensus over changing topologies and over networks with event-triggered communication strategy, however, the dynamic average consensus in the presence of communication time delay has not been addressed. This paper intends to fill this gap as delays are inevitable in the real systems and are known to  cause disruptive   behavior such as network
instability or the network
desynchronization~\cite{WY-JC-GC:08,DH-GK-BKS:10,JIG-GO:17}. 

In this paper, we study the effect of fixed communication delay 
on the continuous-time dynamic average consensus algorithm of~\cite{dps-ros-rmm:05b}
and its discrete-time implementation. 
In the (Laplacian) static average consensus algorithm, 
the reference inputs of the agents enter the algorithm as initial conditions. Therefore, the robustness to delay analysis is focused on identifying the admissible ranges of the delay for which the algorithm is internally exponentially stable. However, in our dynamic average consensus algorithm of interest, instead of the initial conditions, the reference inputs enter the algorithm as an external input. Therefore, in addition to the internal stability analysis, we need to asses the input to output stability and convergence performance. To perform our studies we use a set of time domain approaches. Specifically, in continuous-time, we model our dynamic average consensus~algorithm with communication delay as a delay differential~equation (DDE) and use the characterization of solutions of DDEs and their convergence analysis (c.f. e.g.,~\cite{SN:01,SY-PWN-AGU:10}) to preform our studies.  In discrete-time, we model our algorithm
as a delay difference equation whose solution characterization can be found, for example, in~\cite{JD-DYK:06,EF:14}. For both  of the continuous- and the discrete-time algorithms that we study, we carefully characterize the admissible delay bound over strongly connected and weight-balanced (SCWB) interaction topologies.
  We also characterize the convergence rate and show that in both of the continuous- and the discrete-time cases, the rate is a function of the time delay and network topology. Moreover, we show that the admissible delay bound is an inverse function of the maximum degree of the SCWB digraph, a result that previously only was established for undirected graphs in the case of the static average consensus algorithm. Our convergence analysis also includes establishing a practical bound on the tracking error, and showing that for static signals the algorithms achieve perfect tracking for delays in the admissible bound. We also show for  connected  undirected graphs, the discrete-time  algorithm is guaranteed to tolerate, at least, one step delay. Simulations demonstrate our results and also show that for delays beyond the admissible bound the algorithms under study become unstable. Our robustness analysis approach can be applied to other dynamic average consensus algorithms, as well. For example, in our preliminary work~\cite{HM-SSK:17}, we studied the robustness of the continuous-time dynamic average consensus algorithm  of~\cite{SSK-JC-SM:15-ijrnc} to the communication delay.
 
\emph{Notations}: We let $\reals$, $\realpositive$, $\realnonnegative$, $\mathbb{Z}$, $\mathbb{Z}_{> 0}$ $\mathbb{Z}_{\geq 0}$ and $\mathbb{C}$
denote the set of real, positive real, nonnegative real, integer, positive integer, nonnegative integer, and complex numbers respectively. $\mathbb{Z}_{i}^j$ is the set $\{i,i+1,\cdots,j\}$, where $i,j\in\mathbb{Z}$ and $i<j$. For a $s\in\mathbb{C}$, $\re{s}$ and $\im{s}$ represent its real and imaginary parts, respectively. Moreover, $|s|$ and $\arg{(s)}$ represent its magnitude and argument, respectively, i.e., $|s|=\sqrt{\re{s}^2+\im{s}^2}$ and $\arg(s)=\text{atan2}(\im{s},\re{s})$. When $s\in\real$, $|s|$ is its absolute value. For $\vect{s}\in\reals^d$,
$\|\vect{s}\|=\sqrt{\vect{s}^\top\vect{s}}$ denotes the standard
Euclidean norm. 
We let
$\vect{1}_n$ (resp. $\vect{0}_{n}$) denote the vector of $n$ ones
(resp. $n$ zeros), and denote by $\vect{I}_n$ the $n\times n$ identity
matrix. We let $\PPi_n = \vect{I}_n - \frac{1}{n}
\vect{1}_n\vect{1}_n^\top$.  When clear from the context, we do not
specify the matrix dimensions.  
For a measurable locally essentially bounded function $\vect{u}:\real_{\geq0}\to\real^m$, we define
$\|\vect{u}\|_{\infty}= \text{ess}\sup\{\|\vect{u}(t)\|,~t\geq0\}$. For a discrete-time function $\vect{u}:\mathbb{Z}\to\real^m$, we define $\|\vect{u}\|_{\infty}= \sup\{\|\vect{u}(k)\|,~k\in\mathbb{Z}\}$.
 For $x\in\real$, ceiling of $x$ demonstrated by $\ceil{x}$ is the smallest integer greater than or equal to $x$.
For a given $\mathpzc{d}\in\mathbb{Z}_{\geq0}$ and $\vect{A}\in\real^{n\times n}$, we define delayed exponential of the matrix $\vect{A}$ as
\begin{align}\label{eq::delay-matrix-exponential}
 \ee^{\vect{A}k}_{\mathpzc{d}}\!=\!\sum^{m_k}_{l=0}\binom
  {k-(l-1)\mathpzc{d}}{l}\vect{A}^l,\, m_k=\ceil{\frac{k}{\mathpzc{d}\!+\!1}},\, k\in\mathbb{Z}_{\geq 0}.
\end{align}
In a network of $N$ agents, the aggregate vector of local variables $p^i$, $i\until{N}$, is denoted by $\vect{p} = ({p}^1,\dots,{p}^N)^\top\in\reals^N$. We define  $\rR\in\real^{N\times(N-1)}$ such that
 \begin{align}\label{eq::orthonormal}
 \arraycolsep=0.5pt\def\arraystretch{1}\left[\begin{array}{c}
\!\frac{1}{\sqrt{N}}\vect{1}_N^\top\\
 \rR^\top\end{array}\right]\left[\begin{array}{cc}
\!\frac{1}{\sqrt{N}}\vect{1}_N\,\,&\,
 \rR\end{array}\right]\!=\!\left[\begin{array}{cc}
\!\frac{1}{\sqrt{N}}\vect{1}_N\,&\,\,
 \rR\!\end{array}\right]\left[\begin{array}{c}
\!\frac{1}{\sqrt{N}}\vect{1}_N^\top\\
 \rR^\top\!\end{array}\right]\!=\!\vect{I}_N.
 \end{align}
Since $\rR\rR^\top=\vect{\Pi}_N$ and $\vect{\Pi}_N\vect{\Pi}_N=\vect{\Pi}_N$, we have
\begin{equation}\label{eq::normRy} \|\rR^\top\,\vect{y}\|=\|\vect{\Pi}\,\vect{y}\|,\quad\forall\, \vect{y}\in\real^N.\end{equation}
For a given $z\in\complex$, Lambert $W$ function is defined as the solution of the equation $s\,\ee^s=z$, i.e., $s=W(z)$ (c.f.~\cite{RMC-GHG-DEGH-DJJ-DEK:96,HS-TM:06}). Except for $z=0$ which gives $W(0)=0$, $W$ is a multivalued function with the infinite number of solutions denoted by $W_k(z)$ with $k\in\mathbb{Z}$, where $W_k$ is called the $k^{\text{th}}$ branch of $W$ function. Matlab and Mathematica have functions to evaluate $W_k(z)$. For any $z\in\{z\in\complex\,|\,|z|\leq \frac{1}{\ee}\}$,  we have
\begin{align}\label{eq::W0}
 W_0(z)=\sum\nolimits_{n=1}^{\infty}\frac{(-n)^{n-1}}{n!}z^n.
\end{align}

Next, we review some basic concepts from graph
theory following~\cite{FB-JC-SM:09}. A weighted \emph{digraph}, is a triplet $\GG = (\VV ,\EE,
\vectsf{A})$,~where $\VV=\{1,\dots,N\}$ is the \emph{node set} and
$\EE \subseteq \VV\times \VV$ is the \emph{edge set}, and $\vectsf{A}\in\real^{N\times N}$ is a weighted \emph{adjacency}
matrix such that $ \mathsf{a}_{ij} >0$ if $(i, j) \in\EE$ and $
\mathsf{a}_{ij} = 0$, otherwise.  An edge $(i, j)$ from
$i$ to $j$ means that agent $j$ can send
information to agent $i$. Here,  $i$ is called an
\emph{in-neighbor} of $j$ and $j$ is called an \emph{out-neighbor}
of~$i$.  
A \emph{directed path} is a sequence of nodes
connected by edges. A digraph is \emph{strongly connected} if for
every pair of nodes there is a directed path connecting them.  A weighted digraph is
\emph{undirected} if $ \mathsf{a}_{ij} = \mathsf{a}_{ji}$ for all
$i,j\in\VV$. A \emph{connected} undirected graph is an undirected graph in which any two nodes are connected to each other by paths. The \emph{weighted in-} and
\emph{out-degrees} of a node $i$ are, respectively,
$\mathsf{d}_{\text{in}}^i =\Sigma^N_{j =1} \mathsf{a}_{ji}$ and~$\mathsf{d}_{\text{out}}^i =\Sigma^N_{j =1} \mathsf{a}_{ij}$. Maximum in- and out-degree of digraph are, respectively, $\mathsf{d}_{\text{out}}^{\max}=\max\{\mathsf{d}_{\text{out}}^i\}_{i=1}^{N}$, and  $\mathsf{d}_{\text{in}}^{\max}=\max\{\mathsf{d}_{\text{in}}^i\}_{i=1}^{N}$.
The \emph{(out-) Laplacian} matrix is $\lL =
\vect{\mathsf{D}}^{\text{out}} - \vect{\mathsf{A}}$, where
$\vect{\mathsf{D}}^{\text{out}} =
\Diag{\mathsf{d}_{\text{out}}^1,\cdots, \mathsf{d}_{\text{out}}^N} \in
\reals^{N \times N}$. Note that $\lL\vect{1}_N=\vect{0}$. A
digraph is \emph{weight-balanced} iff at each node $i\in\VV$, the
weighted out-degree and weighted in-degree coincide, (iff $\vect{1}_N^T\lL=\vect{0}$). In a weight-balanced digraph we have $\mathsf{d}_{\text{out}}^{\max}=\mathsf{d}_{\text{in}}^{\max}=\mathsf{d}^{\max}$. 
For a strongly connected and
weight-balanced digraph, $\vectsf{L}$ has one eigenvalue $\lambda_1\!=\!0$, and the rest of the eigenvalues have positive real parts.~Moreover, 
   \begin{equation}\label{eq::RLR}
    0<\hat{\lambda}_2 \vect{I}\leq \rR^\top\Sym{\lL}\rR\leq \hat{\lambda}_N \vect{I},  
     \end{equation} 
    where $\hat{\lambda}_2$ and $\hat{\lambda}_N$, are, respectively, the smallest non-zero eigenvalue and maximum eigenvalue of $\Sym{\lL} =(\vect{\lL}+\vect{\lL}^\top)/2$. 
For  connected  undirected graphs $\{\hat{\lambda}_i\}_{i=1}^N$,  are equal to eigenvalues $\{\lambda_i\}_{i=1}^N$  of $\lL$, therefore, $0<{\lambda}_2 \vect{I}\leq \rR^\top\lL\rR\leq {\lambda}_N \vect{I}$. 

\section{Problem statement}\label{se:problem-statement}
We consider a network of $N$ agents where each agent $i\in\VV$ has access to a time-varying one-sided reference input signal 
\begin{align}\label{eq::input_signal}
\mathsf{r}^i(t)=\begin{cases}{r}^i(t),&t\in\realnonnegative\\
0,&t\in\real_{<0},
\end{cases}
\end{align}
where $r^i:\real_{\geq0}\to\real$ is a measurable locally essentially bounded signal.
The interaction
topology is  a SCWB digraph $\GG$, with communication messages being subject to a fixed common transmission delay $\tau\in\realpositive$. The problem of interest is to enable  each agent's local variable $x^i$ to asymptotically track $\avrg{\mathsf{r}}(t)$ in a distributed manner over this network. Our starting point is the continuous-time dynamic average consensus algorithm
\begin{align}\label{eq::CT}
 &x^i(t)=z^i(t)+\mathsf{r}^i(t),\\
   & \dot{z}^i\!=-\beta\sum\nolimits_{j=1}^N
    \mathsf{a}_{ij}(x^i(t-\tau)-x^j(t-\tau)),\nonumber\\
    &\sum\nolimits_{j=1}^{N}z^j(0)=0, ~z^i(t)=0~\text{for }t\in[-\tau,0), \quad i\in\VV,\nonumber
  \end{align}
which is proposed in~\cite{dps-ros-rmm:05b} in the alternative but equivalent form of  $\dvect{x}=-\beta \lL\vect{x}+\dvectsf{r}$. Algorithm~\eqref{eq::CT} can also be obtained from the proportional dynamic average consensus algorithm of~\cite{RAF-PY-KML:06} when the parameter $\gamma$ in that algorithm is set to zero. 

\begin{thm}(Convergence of~\eqref{eq::CT} over a SCWB digraph when $\tau=0$)\label{lem::Alg_ContiTime}
  Let $\GG$ be a SCWB digraph and assume that there is no communication delay, i.e., $\tau=0$. Let $\|(\vect{I}_N-\frac{1}{N}\vect{1}_N\vect{1}_N^\top)\dvectsf{r}\|_{\infty} \!=\! \gamma\!<\!\infty$.  Then, for
  any $\beta\in\realpositive$, the trajectories of
  algorithm~\eqref{eq::CT} are bounded and~satisfy
  \begin{equation}\label{eq::Alg_D_ultimate_bound} 
    \lim_{t\to\infty} \big| x^i(t)\!-\!\frac{1}{N}\sum\nolimits_{j=1}^N\mathsf{r}^j(t) \big| \leq 
    \frac{\gamma}{ 
      \beta\Hlambda_2},~~~i\in\VV. 
  \end{equation}
  The convergence rate to this error bound is no worse than~$\beta\text{Re}({\lambda}_2)$.
\end{thm}
The proof of this theorem can be found in~\cite{SSK-BVS-JC-RAF-KML-SM-arxiv}.
An iterative form of algorithm~\eqref{eq::CT} with step-size $\delta\in\realpositive$ and integer communication time-step delay $\mathpzc{d}\in\mathbb{Z}_{\geq 0}$, is  
 
 \begin{align}
   &x^i(k)\! =z^i(k) \!+\!\mathsf{r}^i(k),\label{eq::DCDisc}
     \\
     &z^i(k\!+\!1) \!=z^i(k)\!-\delta\beta\,\sum\nolimits_{j=1}^{N}\mathsf{a}_{ij}(x^i(k-\mathpzc{d})-x^j(k-\mathpzc{d})), \nonumber\\
       &\sum\nolimits_{j=1}^{N}z^j(0)=0,~ z^i(k)=0~\text{for }k\in\mathbb{Z}_{-\mathpzc{d}}^{-1},\quad i\in\VV. \nonumber
  \end{align}
  
 \begin{thm}(Convergence of~\eqref{eq::DCDisc} over a SCWB
  digraph when $\mathpzc{d}=0$)\label{lem::Alg_ContiTime}
  Let $\GG$ be a SCWB digraph and assume that there is no communication delay, i.e., $\mathpzc{d}=0$. Let
  $\|(\vect{I}_N-\frac{1}{N}\vect{1}_N\vect{1}_N^\top)\Delta\vectsf{r}\|_{\infty}\!=\! \gamma<\infty$.  Then, for
  any $\beta\in\realpositive$, the trajectories of
  algorithm~\eqref{eq::DCDisc} are bounded and~satisfy
  \begin{equation}\label{eq::Alg_Disc_ultimate_bound} 
    \lim_{k\to\infty} \big| x^i(k)-\frac{1}{N}\sum\nolimits_{j=1}^N\mathsf{r}^j(k) \big| \leq 
    \frac{\gamma\delta}{\beta\Hlambda_2},
  \end{equation}
  for $i\in\VV$, provided
  $\delta\in(0,\beta^{-1}(\dM)^{-1})$.
\end{thm} 
The proof of this theorem can be deduced from the results in~\cite{SSK-BVS-JC-RAF-KML-SM-arxiv}, and is omitted for brevity.
Note that the initialization condition~$\sum_{j=1}^Nz^j(0)=0$ in algorithm~\eqref{eq::CT} and its discrete-time implementation can be easily satisfied by assigning $z^i(0)=0$, $i\in\VV$. Initialization conditions appear in other dynamic average consensus algorithms, as well~\cite{ROS-JSS:05,MZ-SM:08a,HB-RAF-KML:10,SSK-JC-SM:15-ijrnc,FC-YC-WR:12}.

Our objective is to characterize the admissible communication delay ranges for $\tau\in\realnonnegative$ and $\mathpzc{d}\in\mathbb{Z}_{\geq 0}$ for the  dynamic average consensus algorithms above. We also want to study the effect of a fixed communication delay on the tracking error bound and the rate of convergence of these algorithms. By admissible delay value we mean values of delay for which the algorithms stay internally exponentially stable. A short review of the definition of the exponential stability of linear delayed systems is provided in the appendix.

\section{Convergence and stability analysis in the presence of a constant communication delay}\label{sec::main}
In this section, we study the stability and convergence properties of algorithm~\eqref{eq::CT} and its discrete-time implementation~\eqref{eq::DCDisc} in the presence of a constant communication~delay. 

\subsection{Continuous-time case}
For convenience in analysis, we start our study by applying the change of variable  (recall $\rR$ from~\eqref{eq::orthonormal})
\begin{align}\label{eq::change-var}
\begin{bmatrix}
p_1\\ \vect{p}_{2:N}\end{bmatrix}
=
\begin{bmatrix}\frac{1}{\sqrt{N}}\vect{1}_N^\top\\\rR^\top\end{bmatrix}
(\vect{x}-\frac{1}{N}\sum\nolimits_{j=1}^N\mathsf{r}^j\vect{1}_N)
\end{align}
to represent  algorithm~\eqref{eq::CT} in the equivalent compact form 
\begin{subequations}\label{eq::DEvent_Alg_Separated}
  \begin{align}
    \dot{p}_{1}(t) &=0,\label{eq::DEvent_Alg_Separated-b}
    \\
   \dvect{p}_{2:N}(t)& =\!\vectsf{H}\, 
    \vect{p}_{2:N}(t-\tau)+\vect{R}^\top\dvectsf{r}(t)
    , \label{eq::DEvent_Alg_Separated-d}
  \end{align}\end{subequations}
  
where
\begin{align}\label{eq::H}
 \vectsf{H}=-\beta\,\rR^\top \vectsf{L}\rR.
\end{align}

Because of~\eqref{eq::orthonormal}, $\|\vect{x}(t)-\frac{1}{N}\sum\nolimits_{j=1}^N\mathsf{r}^j(t)\vect{1}_N\|^2=\|\vect{p}(t)\|^2=\|\vect{p}_{2:N}(t)\|^2+|p_1(t)|^2$. Therefore, we can~write
\begin{align}
\lim_{t\to\infty}\!\big|x^i(t)&-\frac{1}{N}\sum\limits_{j=1}^N\mathsf{r}^j(t)\big|^2\!\leq\lim_{t\to\infty} \! \big\|\vect{x}(t)-\frac{1}{N}\sum\limits_{j=1}^N\mathsf{r}^j(t)\vect{1}_N\big\|^2\nonumber\\
&=\lim_{t\to\infty}\Big(\|\vect{p}_{2:N}(t)\|^2+|p_1(t)|^2\Big),\quad i\in\VV.\label{eq::ult-tracking-temp}
\end{align}
For algorithm~\eqref{eq::CT}, using $\sum_{i=1}^Nz^i(0)=0$, we obtain 
\begin{align}\label{eq::initial_q_1}
\!p_1(0)\!=\!\frac{1}{\sqrt{N}}\vect{1}_N^\top(\vect{x}(0)\!-\,\vectsf{r}(0))\!=\!\frac{1}{\sqrt{N}}\sum\nolimits_{i=1}^Nz^i(0)=0.
\end{align}
Therefore,~\eqref{eq::DEvent_Alg_Separated-b} gives $p_1(t)=0$, $t\in\real_{\geq0}$.
To establish an upper bound on the tracking error of each agent, we need to obtain an upper bound on  $\|\vect{p}_{2:N}(t)\|$, as well. Since~\eqref{eq::DEvent_Alg_Separated-d} is a DDE system with the system matrix $\vectsf{H}$ and a delay free input $ \vect{R}^\top \dvectsf{r}$, the admissible delay bound for~\eqref{eq::DEvent_Alg_Separated-d} and subsequently, for the dynamic average consensus algorithm~\eqref{eq::CT},  is determined by the delay bound for the zero input dynamics of ~\eqref{eq::DEvent_Alg_Separated-d}, i.e.,
  \begin{align}\label{eq::p-zeroInput}
    \dvect{p}_{2:N}(t) & =\! \vectsf{H}\,
    \vect{p}_{2:N}(t-\tau).
    \end{align}
 Note that~\eqref{eq::p-zeroInput} is the Laplacian average consensus algorithm with the zero eigenvalue separated. For  connected  undirected graphs,~\cite{ROS-RMM:04} used the Nyquist criterion to characterize the admissible delay bound. Here, to obtain the admissible delay range of~\eqref{eq::p-zeroInput} over SCWB~digraphs, we use a result based on the characteristic equation analysis for linear delay systems. 

\begin{lem}(Admissible range of $\tau$ for~\eqref{eq::p-zeroInput} over SCWB digraphs
)\label{lem::admissible_tau_consens}
Let $\GG$ be a SCWB digraph. Then, for any $\tau\in[0,\bar{\tau})$ the zero-input dynamics~\eqref{eq::p-zeroInput} is exponentially stable if and only~if
\begin{align}\label{eq::delay_bound_consens}
\bar{\tau}=\min{\Big\{\tau\in\realpositive\,\Big|\,\tau=\frac{|\text{atan}(\frac{\re{\lambda_i}}{\im{\lambda_i}})|}{\beta\,|\lambda_i|}, ~i\in\{2,\cdots,N\}
}
\Big\},
\end{align}
where $\{\lambda_i\}_{i=2}^N$ are the non-zero eigenvalues of $\vectsf{L}$.
Also, for  connected  undirected graphs, we have $\bar{\tau}=\frac{\pi}{2\,\beta\lambda_N}$. 
\end{lem}
\begin{proof}
Recall that for the SCWB digraphs, $\vectsf{H}\in\real^{(N-1)\times(N-1)}$ is a Hurwitz matrix whose eigenvalues are $\{-\beta\,\lambda_i\}_{i=2}^N$. Then, our proof is the straightforward application of~\cite[Theorem~1]{MB:87}, which states that a linear delayed system  $\dvect{\chi}=\vect{A}\vect{\chi}(t-\tau)$, is exponentially stable if and only if $\vect{A}\in\real^{n\times n}$ is Hurwitz and  $0\leq\tau<\frac{|\text{atan}(\frac{\re{\mu_i}}{\im{\mu_i}})|}{|\mu_i|}, ~i\in\{1,\cdots,n\}$, where $\mu_i$ is the $i^{\text{th}}$ eigenvalue of $\vect{A}$. For  connected  undirected graphs, we have $\{\beta\,\lambda_i\}_{i=2}^N\subset\real_{>0}$, hence $|\text{atan}(\frac{\re{\lambda_i}}{\im{\lambda_i}})|\!=\!\frac{\pi}{2}$, and $\max \{\beta\,|\lambda_i|\}_{i=2}^N\!=\!\beta\lambda_N$. Therefore $\bar{\tau}$ in~\eqref{eq::delay_bound_consens} is equal to $\frac{\pi}{2\,\beta\lambda_N}$.
\end{proof}

For delays in the admissible bound characterized in Lemma~\ref{lem::admissible_tau_consens}, the zero-input dynamics~\eqref{eq::p-zeroInput} is exponentially stable. Therefore, there are $\kappaa_{\,\,\tau}\in\realpositive$ and  $\rho_{\tau}\in\realpositive$ such that (see Definition~\ref{def::expen-stabil-continuous-delay})
\begin{align}\label{eq::expon_stabili}\|\vect{P}_{2:N}(t)\|\leq \kappaa_{\,\,\tau}\text{e}^{-\rho_{\tau}\,t}\underset{\eta\in[-\tau,0]}{\text{sup}}\|\vect{p}_{2:N}(\eta)\|, \quad t\in\realnonnegative.\end{align}
Next, we use the results in~\cite{SD-JN-AGU:11} to obtain $\kappaa_{\,\,\tau}$ and $\rho_{\tau}$.

\begin{lem}(Exponentially decaying upper bound for  $\|\vect{p}_{2:N}(t)\|$ in~\eqref{eq::p-zeroInput})\label{lem::expon_bound_zero_sys}Let $\tau\in[0,\bar{\tau})$, where $\bar{\tau}$ is given in~\eqref{eq::delay_bound_consens}. Then, 
the trajectories of zero-input dynamics~\eqref{eq::p-zeroInput} over SCWB digraphs satisfy the exponential bound in~\eqref{eq::expon_stabili} with 
\begin{subequations}\label{eq::rate_gain_characterization}
\begin{align}
\rho_\tau &=\frac{1}{\tau}\max\{\text{Re}(W_0(-\beta\lambda_i\tau)\}_{i=2}^N,\\
\kappaa_{\,\,\tau}&=\max({K_1},{K_2})+\max({K_3},{K_4}),
\end{align}
\end{subequations}
where $\{\lambda_i\}_{i=2}^N$ are the none zero eigenvalues of $\lL$ and $K_1,K_2, K_3,K_4$ are gains that can be computed based on systems matrices (the closed form expressions are omitted for brevity, please see ~\cite[Theorem~1]{SD-JN-AGU:11}).
\end{lem}
\begin{proof}
The proof is the direct application of~\cite[Theorem~1]{SD-JN-AGU:11} which is omitted to avoid repetition. In applying~\cite[Theorem~1]{SD-JN-AGU:11}, one should recall that system matrix $\vectsf{H}$ in~\eqref{eq::p-zeroInput} is Hurwitz, its nullity is zero and, its eigenvalues are $\{-\beta\lambda_i\}_{i=2}^N$.\end{proof}
Using our preceding results, our main theorem below establishes admissible delay bound, an ultimate tracking bound, and the rate of convergence to this error bound for the distributed average consensus algorithm~\eqref{eq::CT} over SCWB digraphs.

\begin{thm}(Convergence of~\eqref{eq::CT} over SCWB digraphs in the presence of communication delay)\label{lem::Alg_delay}
Let $\GG$ be a SCWB digraph with communication delay in $\tau\in[0,\bar{\tau})$ where $\bar{\tau}$ is given in~\eqref{eq::delay_bound_consens}. Let $\|(\vect{I}_N-\frac{1}{N}\vect{1}_N\vect{1}_N^\top)\dvectsf{r}\|_{\infty}=\gamma<\infty$.  Then, for
  any $\beta\in\realpositive$, the trajectories of algorithm~\eqref{eq::CT} are bounded and satisfy

\begin{equation}\label{eq::Alg_D_ultimate_bound} 
    \lim_{t\to\infty} \Big| x^i(t)-\avrg{\mathsf{r}}(t) \Big| \leq 
    \frac{\gamma\kappaa_{\,\,\tau}}{ 
      \rho_{\tau}},
  \end{equation}
   for  $i\in\VV$. Here,  $\rho_{\tau},\kappaa_{\,\,\tau}\in\realpositive$ are given by~\eqref{eq::rate_gain_characterization}.
  The rate of convergence to this error neighborhood is no worse than~$\rho_{\tau}$.
\end{thm}

\begin{proof}
Recall the equivalent representation~\eqref{eq::DEvent_Alg_Separated} of~\eqref{eq::CT} and~also \eqref{eq::initial_q_1}. We have already shown that $p_1(t)=0$. for $t\in\real_{\geq0}$. 
 Since under the given initial condition we have $\vect{p}_{2:N}(t)  = \vect{0}_{N-1}$ for $t\in[-\tau,0)$, the trajectories $t\mapsto\vect{p}_{2:N}$ of~\eqref{eq::DEvent_Alg_Separated-d} are 
 \begin{align*}\vect{p}_{2:N}(t)=\,&\vect{\Phi}(t)\,\vect{p}_{2:N}(0)-\!\!\int_{0}^t\vect{\Phi}(t-\zeta)\,\rR^\top\,\dvectsf{r}(\zeta)\text{d}\zeta,\end{align*}
where $\Phi(t-\zeta)=\sum\nolimits_{k=-\infty}^{k=\infty}\text{e}^{\vect{S}_k(t-\zeta)}\vect{C}_k$ with  $\vect{S}_k=\frac{1}{\tau}W_k(\vectsf{H}\tau)$, and each coefficient $\vect{C}_k$ depending on $\tau$ and $\vectsf{H}$ (c.f.~\cite{SY-PWN-AGU:10} for details). 
Then, we can write
\begin{align}\label{eq::upperbound2}
  \|\vect{p}_{2:N}(t)\|\leq&\, \|\vect{\Phi}(t)\,\vect{p}_{2:N}(0)\|\,+
   \gamma\int_{0}^t\|\vect{\Phi}(t-\zeta)\|\text{d}\zeta.
\end{align}  
Here, we used~\eqref{eq::normRy} to write $\|\vect{R}^\top\,\dvectsf{r}\|_{\infty}=\|\vect{\Pi}\,\dvectsf{r}\|_{\infty}=  \gamma $. Because for $\tau\in[0,\bar{\tau})$ the zero-input dynamics of~\eqref{eq::DEvent_Alg_Separated-d} is exponentially stable, by invoking the results of Lemma~\ref{lem::expon_bound_zero_sys} we can deduce that the trajectories of the zero-input dynamics of~\eqref{eq::DEvent_Alg_Separated-d} for $t\in\realnonnegative$
satisfy $\|\vect{\Phi}(t) \vect{p}_{2:N}(0)\|\leq \kappaa_{\,\,\tau}\text{e}^{-\rho_{\tau}\,t}\|\vect{p}_{2:N}(0)\|$, where $\rho_{\tau}$ and $\kappaa_{\,\,\tau}$ are described in the statement. Here, we used $\underset{t\in[-\tau,0]}{\text{sup}}\|\vect{p}_{2:N}(t)\|=\|\vect{p}_{2:N}(0)\|$, which holds because of the given initial conditions. Because this bound has to hold for any initial condition including those satisfying $\|\vect{p}_{2:N}(0)\|=1$, we can conclude that $\|\vect{\Phi}(t)\|\leq \kappaa_{\,\,\tau}\text{e}^{-\rho_{\tau}\,t}$, for $t\in\realnonnegative$ (recall the definition of the matrix norm $\|\vect{\Phi}(t)\|=\sup\big\{\|\vect{\Phi}(t)\,\vect{p}_{2:N}(0)\|\,\big|\, \|\vect{p}_{2:N}(0)\|=1\big\}$).
 Therefore, from~\eqref{eq::upperbound2} we obtain $\|\vect{p}_{2:N}(t)\| \leq
   {\kappaa}_{\,\,\tau}\text{e}^{-\rho_{\tau}\,t}\,\|\vect{p}_{2:N}(0)\|+\gamma \int_{0}^t\kappaa_{\,\,\tau}\text{e}^{-\rho_{\tau}\,(t-\zeta)}\,\text{d}\zeta.$
 Then, using $\int_{0}^t\text{e}^{-\rho_{\tau}(t-\zeta)}\mathpzc{d}\zeta=\frac{1}{\rho_{\tau}}(1-\text{e}^{-\rho_{\tau}t}),$
we can write 
\begin{align}\label{eq::bound_p2N}
  \!\!\!\!\! \|\vect{p}_{2:N}(t)\| \leq& 
   {\kappaa}_{\,\,\tau}\text{e}^{-\rho_{\tau}\,t}\,\|\vect{p}_{2:N}(0)\|+\frac{\kappaa_{\,\,\tau}}{\rho_{\tau}}(1-\text{e}^{-\rho_{\tau}t})\gamma
 \end{align}
 The boundedness of the trajectories of~\eqref{eq::bound_p2N} and the correctness of~\eqref{eq::upperbound2} follow from~\eqref{eq::ult-tracking-temp},~\eqref{eq::p-zeroInput} and~\eqref{eq::bound_p2N}. Moreover, the rate of convergence is also $\rho_{\tau}$.
\end{proof}

\medskip
Observe that 
\begin{align*}
\lim_{\tau\to0}\rho_\tau&=\lim_{\tau\to0}-\frac{1}{\tau} \max\{\re{W_0(-\beta\lambda_i\,\tau)}\}_{i=2}^N\\&= -\max\{\re{-\beta\lambda_i}\}_{i=2}^N=\beta\re{\lambda_2}.   
\end{align*}
In other words, as $\tau\to0$, the rate of convergence of algorithm~\eqref{eq::CT} converges to its respective value of the delay free implementation given in Theorem~\ref{lem::Alg_ContiTime}. Moreover, note that ~\eqref{eq::Alg_D_ultimate_bound} indicates that if the reference inputs are static, which means $\|\rR^\top\,\dvectsf{r}\|_{\infty}=\gamma=0$,  for any admissible delay   algorithm~\eqref{eq::CT} converges to  the exact average of the reference inputs.

Next, we establish a relationship between the upper bound of the admissible delay bound and the maximum degree of the communication graph. For connected  undirected graphs as $\lambda_N\leq 2 \,\dM$ (see~\cite{ROS-RMM:04}), the admissible delay range that is identified in Lemma~\eqref{lem::admissible_tau_consens}  satisfies $[0,\frac{\pi}{4\,\beta\dM})\subseteq [0,\bar{\tau})$. We extend the results to SCWB digraphs.

\begin{lem}(Admissible range of $\tau$ for~\eqref{eq::p-zeroInput} in terms of maximum degree of the digraph)\label{cor::admissible_tau_consens-degree}
Let $\GG$ be a SCWB digraph. Then, $\bar{\tau}$ in~\eqref{eq::delay_bound_consens} satisfies 
\begin{align}\label{eq::delay_bound_consens-out-deg}
\bar{\tau}\geq  1/(2\,\beta\,\dM).
\end{align}
\end{lem}
\begin{proof}
Recall that $\{\lambda_i\}_{i=2}^N$ are the non-zero eigenvalues of $\vectsf{L}$. By invoking the Gershgorin circle theorem~\cite{RAH-CRJ:85}, we can write $(\re{\lambda_i}-\dout^i)^2+\im{\lambda_i}^2\leq \dout^i$, $i\in\{2,\cdots,N\}$.
Given that $0<\dout^i\leq \dM$, then for any $\lambda_i$ we have $(\re{\lambda_i}-\dM)^2+\im{\lambda_i}^2\leq (\dM)^2,~i\in\!\{2,\cdots,N\}$,
which is equivalent to 
$ \re{\lambda_i}^2+\im{\lambda_i}^2\leq 2\re{\lambda_i}\,\dM$. 
 Hence, we can write 
\begin{align*}
\frac{1}{2\dM}\leq &\frac{\re{\lambda_i}}{\re{\lambda_i}^2+\im{\lambda_i}^2}=\frac{\re{\lambda_i}}{\sqrt{\re{\lambda_i}^2+\im{\lambda_i}^2}}\frac{1}{|\lambda_i|}.
\end{align*}
Let $\phi=\text{asin}(\frac{\re{\lambda_i}}{\sqrt{\re{\lambda_i}^2+\im{\lambda_i}^2}})$. Since $\text{tan}(\text{asin}(x))=\frac{x}{\sqrt{1-x^2}}$ holds for any $x\in\real$, one can yield that $\phi=\pm\text{atan}(\frac{\re{\lambda_i}}{\im{\lambda_i}})$. 
Also, for any $\phi\in\reals$ we have $\text{sin}(\phi)\leq|\phi|$. Thus, we have
\begin{align*}
 \frac{1}{2\beta\dM}\leq &\frac{\re{\lambda_i}}{\beta(\re{\lambda_i}^2+\im{\lambda_i}^2)}\leq\frac{|\text{atan}(\frac{\re{\lambda_i}}{\im{\lambda_i}}|}{\beta|\lambda_i|},  \end{align*}
proving that ~\eqref{eq::delay_bound_consens-out-deg} is a sufficient condition for~\eqref{eq::delay_bound_consens}.
\end{proof}
\medskip

The inverse relation between the maximum admissible delay and the maximum degree of the communication topology is aligned with the intuition. One can expect that the more links to arrive at some agents of the network, the more susceptible the convergence of the algorithm will be to the larger delays. 

\vspace{-0.1in}

\subsection{Discrete-time case}
For convenience in stability analysis, we use the change of variable~\eqref{eq::change-var}, to represent~\eqref{eq::DCDisc} in the equivalent form 
\begin{subequations}\label{eq::Dynamic_res_dis}
\begin{align}
\!\!\!\!\!p_1(k\!+\!1)&\!=\!p_1(k)
,\label{eq::Dynamic_res_dis-p1}\\
\!\!\vect{p}_{2:N}(k\!+\!1)&\!=\!\vect{p}_{2:N}(k)\!+\!\delta\,\vectsf{H}\,\vect{p}_{2:N}(k\!-\!\mathpzc{d}) \!-\!\delta \rR^\top\Delta\vectsf{r}(k).\label{eq::discrete-p-2N}
\end{align}
\end{subequations}
where $\vectsf{H}$ is defined in~\eqref{eq::H} and $\Delta\mathsf{r}^i(k)=\mathsf{r}^i(k+1)-\mathsf{r}^i(k)$. In what follows, We conduct our analysis for $\delta\in(0,\beta^{-1}(\dM)^{-1})$, a stepsize for which the delay free algorithm is stable, see Theorem~\ref{lem::Alg_ContiTime}.

Similar to the continuous-time case, the tracking error of  algorithm~\eqref{eq::DCDisc} for each agent $i\in\VV$ satisfies  
\begin{align}\label{eq::ult-tracking-disc-temp}
\lim_{k\to\infty}\!\Big|x^i(k)-\frac{1}{N}\sum_{j=1}^N \mathsf{r}^j(k)\Big|^2\!\leq&\lim_{k\to\infty}\Big(\|\vect{p}_{2:N}(k)\|^2\!+\!|p_1(k)|^2\Big).
\end{align}
Under the assumption that $\sum_{i=1}^{N}z^i(0)=0$, which gives~\eqref{eq::initial_q_1},~\eqref{eq::Dynamic_res_dis-p1} results in $p_1(k)=p_1(0)=0$ for all $k\in\mathbb{Z}_{\geq0}$.
To establish an upper bound on the tracking error of each agent, we need to obtain an upper bound on  $\|\vect{p}_{2:N}(k)\|$.
Following~\cite[Thoerems 3.1 and 3.5]{JD-DYK:06}, the trajectories of~\eqref{eq::discrete-p-2N}  are described by (recall~\eqref{eq::delay-matrix-exponential})
\begin{align}\label{eq::p2-traj-k}
\vect{p}_{2:N}(k)=&\,\sum\nolimits_{j=-\mathpzc{d}+1}^{0}\ee^{\delta\,\vectsf{H}\,(k-\mathpzc{d}-j)}_{\mathpzc{d}}\Delta\vect{p}_{2:N}(j-1)-\nonumber\\
&\delta\sum\nolimits^{k-1}_{j=0}\ee^{\delta\,\vectsf{H}(k-j-1-\mathpzc{d})}_{\mathpzc{d}} \rR^\top\Delta\vectsf{r}(j),~k\in\mathbb{Z}_{\geq 0}.\end{align}

We start our analysis by characterizing the admissible ranges of time-step delay $\mathpzc{d}$ for the zero-input dynamics of~\eqref{eq::discrete-p-2N}, i.e.,
\begin{align}\label{eq::discrete-p-2N-zeroinput}
    \vect{p}_{2:N}(k+1)=&\vect{p}_{2:N}(k)+\delta\vectsf{H}\,\vect{p}_{2:N}(k-\mathpzc{d}).
\end{align}

\begin{lem}[Admissible range of $\mathpzc{d}$ for~\eqref{eq::discrete-p-2N-zeroinput} over SCWB digraphs]\label{lem::admissible-d}
Let $\GG$ be a SCWB digraph. Assume that $\delta\in(0,(\beta\,\dM)^{-1})$. Then, for any $\mathpzc{d}\in[0,\bar{\mathpzc{d}})$ the zero-input dynamics~\eqref{eq::discrete-p-2N-zeroinput} is exponentially stable if and only if 
\begin{align}\label{eq::stability con_dis}
\bar{\mathpzc{d}}\!=\!\min\Big\{d\in\mathbb{Z}_{\geq 0}\big|&\,d>\hat{d},~~\hat{d}=\frac{1}{2}\big(\frac{\pi-2|\arg(\lambda_i)|}{2\arcsin(\frac{\beta|\lambda_i|\delta}{2})}\!-\!1\big),~\nonumber\\ ~
&\,\,i\in\{2,\cdots,N\}\Big\},
\end{align}
where $\{\lambda_i\}_{i=2}^N$ are the non-zero eigenvalues of $\vectsf{L}$.
 For  connected  undirected graphs, $\hat{d}$ in~\eqref{eq::stability con_dis} is  $\hat{d}=\frac{1}{2}\big(\frac{\pi}{2\arcsin(\frac{\beta\,\lambda_i\,\delta}{2})}-1\big)$.
\end{lem}
\begin{proof}
When $\mathpzc{d}=0$, by knowing that eigenvalues of $\vect{I}+\delta\,\vectsf{H}$ are $\{1-\beta\delta\lambda_i\}_{i=2}^N$~\cite{SSK-JC-SM:15-ijrnc} shows that  for ~\eqref{eq::discrete-p-2N-zeroinput} to be exponentially stable over SCWB digraphs, $\delta$ has to satisfy~$\delta\in(0,(\beta\,\dM)^{-1})$. 
To characterize an upper bound for admissible $\mathpzc{d}\in\mathbb{Z}_{> 0}$, we invoke~\cite[Theorem~1]{ISL:05} which states that the discrete-time delayed system~\eqref{eq::discrete-p-2N-zeroinput} with a fixed delay $\mathpzc{d}\in\mathbb{Z}_{> 0}$ is exponentially stable iff eigenvalues of $\delta\,\vectsf{H}$ lie inside the region of complex plane enclosed by the curve 
\begin{align}\label{eq::stab_zone_disc}
 \Gamma=\big\{z\in\complex|z=2\,\ii\,\sin(\frac{\phi}{2\mathpzc{d}+1})\text{e}^{\ii\phi} , -\frac{\pi}{2}\leq\phi\leq\frac{\pi}{2}\big\}.
\end{align} 
First, we consider $0<\phi<\frac{\pi}{2}$. If we write $z$ in~\eqref{eq::stab_zone_disc} as $z=2\ii\sin(\frac{\phi}{2\mathpzc{d}+1})\ee^{\ii\phi}=2\sin(\frac{\phi}{2\mathpzc{d}+1})\ee^{\ii(\frac{\pi}{2}+\phi)}$, then eigenvalue $-\delta\beta\lambda_i=\delta\,\beta\,|\lambda_i|\ee^{(\pi+\arg(\lambda_i))\ii}$, $i\in\{2,\cdots,N\}$, of $\delta\,\vectsf{H}$ lies inside $\Gamma$ if and only if $\pi+\arg(\lambda_i)=\phi+\frac{\pi}{2}$ and $\beta\delta|\,\lambda_i|<2\,\sin(\frac{\phi}{2\mathpzc{d}+1})$,
in which $-\frac{\pi}{2}<\arg(\lambda_i)<0$, thus,
$
    \beta\delta|\lambda_i|<2\sin(\frac{\frac{\pi}{2}+\arg(\lambda_i)}{2\mathpzc{d}+1}).
$
 Similarly, for $-\frac{\pi}{2}<\phi<0$, (i.e. $0<\arg(\lambda_i)<\frac{\pi}{2}$, ) we obtain $\beta\delta|\lambda_i|<2\sin(\frac{\frac{\pi}{2}-\arg(\lambda_i)}{2\mathpzc{d}+1})$. Therefore, for~\eqref{eq::discrete-p-2N-zeroinput} to be asymptotically stable, we have 
\begin{align*}
\mathpzc{d}<\frac{1}{2}\big(\frac{\pi-2|\arg(\lambda_i)|}{2\arcsin(\frac{\delta\,\beta\,|\lambda_i|}{2})}-1\big),\quad \forall i\in\{2,\cdots,N\}.
\end{align*}
For  connected  undirected graphs,  the eigenvalues of $\delta\,\vectsf{H}$, i.e., $\{-\beta\delta\lambda_i\}_{i=2}^N$ are real, hence $\arg(\lambda_i)=0$. As a result
for~\eqref{eq::discrete-p-2N-zeroinput} to be exponentially stable over  connected  undirected graphs, we obtain $\mathpzc{d}<\frac{1}{2}\big(\frac{\pi}{2\arcsin(\frac{\delta\,\beta\,\lambda_i}{2})}-1\big)$.
This completes our proof.
\end{proof}

Using our preceding results and the auxiliary  Lemma~\ref{lem::exponential_rate_bound_dc} that we presented in Appendix, our main theorem below establishes admissible delay bound, an ultimate tracking bound, and the rate of convergence to the error bound for the discrete-time  algorithm~\eqref{eq::DCDisc} over SCWB digraphs.

\begin{thm}(Convergence of~\eqref{eq::DCDisc} over SCWB digraphs in the presence of communication delay)
\label{thm::Alg_delay}
Let $\GG$ be a SCWB digraph with with communication delay in $d\in[0,\bar{\mathpzc{d}})$ where $\bar{\mathpzc{d}}$ is given in~\eqref{eq::stability con_dis}. Let  $\|(\vect{I}_N-\frac{1}{N}\vect{1}_N\vect{1}_N^\top)\Delta\vectsf{r}\|_{\infty}=\gamma<\infty$.  Then, for
  any $\beta\in\realpositive$ and $\delta\in(0,\beta^{-1}(\dM)^{-1})$, the trajectories of algorithm~\eqref{eq::DCDisc} are bounded and satisfy
\begin{equation}\label{eq::Alg_D_ultimate_bound_dis} 
    \lim_{k\to\infty} \Big| x^i(k)\!-\!\frac{1}{N}\sum\nolimits_{j=1}^N\!\mathsf{r}^j(k) \Big| \leq \,\frac{\gamma\delta\,\bar{\kappaa}_{\mathpzc{d}}}{1-\bar{\omega}_{\mathpzc{d}}}, 
  \end{equation}
for $~i\in\VV$, where 
  \begin{align}\label{eq::LMI}
  (\bar{\omega}_\mathpzc{d},\bar{\kappaa}_{\,\,\mathpzc{d}},\Bvect{Q})=&\underset{\omega_{\mathpzc{d}},\kappa_{\mathpzc{d}},\vect{Q}}{\argmin}\,\,\omega_{\mathpzc{d}}^2,\quad\text{subject to }~\eqref{eq::expon-condi-delay},
\end{align}
with $\vect{A}_{\text{aug}}\!=\!\begin{bmatrix}\vect{0}_{\mathpzc{d}(N-1)\times (N-1)}&\vect{I}_{\mathpzc{d}(N-1)}\\
\delta\,\vectsf{H}&\begin{bmatrix}\vect{0}_{(N-1)\times (\mathpzc{d}-1)(N-1)}&\vect{I}_{N-1}\end{bmatrix}\end{bmatrix}$.
  \end{thm}
\begin{proof}
Consider~\eqref{eq::Dynamic_res_dis}, the equivalent representation of  algorithm~\eqref{eq::DCDisc}. We have already established that $p_1(k)=0$.  Next, we establish an upper bound on trajectories $k\mapsto\|\vect{p}_{2:N}\|$ for $k\in\mathbb{Z}_{\geq0}$. For a $\mathpzc{d}$ in the admissible range defined in the statement, Lemma~\ref{lem::admissible-d} guarantees that the zero dynamic of \eqref{eq::discrete-p-2N}, i.e.,~\eqref{eq::discrete-p-2N-zeroinput}, is exponentially stable. Therefore, invoking Lemma~\ref{lem::exponential_rate_bound_dc}, there always exist a positive definite $\vect{Q}\in\real^{(\mathpzc{d}+1)(N-1)\times(\mathpzc{d}+1)(N-1)}$ and scalars $0<\omega_{\mathpzc{d}}<1$ and $\kappaa_{\,\,\mathpzc{d}}\geq 1$ that satisfy~\eqref{eq::expon-condi-delay} for $\vect{A}_{\text{aug}}$ as defined in the statement. The smallest $\omega_{\mathpzc{d}}$ can be obtained from the matrix inequality optimization problem~\eqref{eq::LMI}. Then, using the results of Lemma~\ref{lem::exponential_rate_bound_dc}, we have the guarantees that the solutions of the zero-state dynamics~\eqref{eq::discrete-p-2N-zeroinput}, for $k\in\mathbb{Z}_{\geq0}$ satisfy $\|\vect{p}_{2:N}(k)\|=\|\text{e}_{\mathpzc{d}}^{\delta\,\vectsf{H} (k-\mathpzc{d})}\vect{p}_{2:N}(0)\|\leq \bar{\kappaa}_{\,\,\mathpzc{d}}\,\bar{\omega}_{\mathpzc{d}}^k\,\|\vect{p}_{2:N}(0)\|$. Because this bound holds for any $\vect{p}_{2:N}(0)$ including those satisfying $\|\vect{p}_{2:N}(0)\|=1$, we can conclude that (recall the definition of a norm of a matrix) $\|\text{e}_{\mathpzc{d}}^{\delta\,\vectsf{H} (k-\mathpzc{d})}\|\leq \bar{\kappaa}_{\,\,\mathpzc{d}}\,\bar{\omega}_{\mathpzc{d}}^k$, for all $k\in\mathbb{Z}_{\geq0}$.
Consequently, from~\eqref{eq::p2-traj-k}, along with $\|\vect{R}^\top\,\Delta\vectsf{r}\|_{\infty}=\|\vect{\Pi}\,\Delta\vectsf{r}\|_{\infty}=\gamma$ (recall~\eqref{eq::normRy}), we can write 
\begin{align}\label{upper-bound_p2}
&\|\vect{p}_{2:N}(k)\|\leq \bar{\kappaa}_{\,\,\mathpzc{d}}\,\bar{\omega}_{\mathpzc{d}}^{k}\, \|\vect{p}_{2:N}(0)\|+\delta\gamma\,\sum\nolimits^{k-1}_{j=0}\bar{\kappaa}_{\,\,\mathpzc{d}}\,\bar{\omega}_{\mathpzc{d}}^{(k-j-1)}.\nonumber\\
\end{align}
Note that $\sum^{k-1}_{j=0}\,\bar{\omega}_{\mathpzc{d}}^{(k-j-1)}=\frac{1-\bar{\omega}_{\mathpzc{d}}^{k}}{1-\bar{\omega}_{\mathpzc{d}}}$.  
As a result, when $k\to\infty$ we obtain (recall that $0<\omega_{\mathpzc{d}}<1$)
Moreover,   $\lim_{k\to\infty}\sum\nolimits^{k-1}_{j=0}\bar{\omega}_{\mathpzc{d}}^{(k-j-1)}=(1-\bar{\omega}_{\mathpzc{d}})^{-1}$.
Therefore, as $k\to\infty$, from~\eqref{upper-bound_p2}, we can conclude that
$\lim_{k\to\infty}\|\vect{p}_{2:N}(k)\|\leq\frac{\bar{\kappaa}_{\,\,\mathpzc{d}}\,\delta\,\gamma}{(1-\bar{\omega}_{\mathpzc{d}})}.
$
\end{proof}

Note that the optimization problem~\eqref{eq::LMI} is a convex linear matrix inequality (LMI) in variables $(\omega_{\mathpzc{d}}^2,\frac{1}{\kappaa_{\,\,\mathpzc{d}}},\vect{Q})$ and can be solved using efficient LMI solvers. Also notice that the tracking error bound~\eqref{eq::Alg_D_ultimate_bound_dis} implies that if the local reference signals are static, i.e., $\|\rR^\top\,\Delta\vectsf{r}\|_{\infty}= \gamma=0$, for any admissible delay, algorithm~\eqref{eq::DCDisc} converges to  the exact average of the reference inputs.

Because in  connected  undirected graphs all the non-zero eigenvalues of Laplacian matrix are real and satisfy $0<\delta\beta\lambda_i<1$, $i\in\{2,\cdots,N\}$,  algorithm~\eqref{eq::DCDisc} is guaranteed to tolerate, at least, one
step delay as  $\frac{\pi}{2\arcsin(\frac{\delta\,\beta\,\lambda_i}{2})}>3$. We close this section with establishing a relationship between admissible delay bound of~\eqref{eq::DCDisc} and the maximum  degree  of the~communication  graph. Here also similar to the continuous-time case, this relationship is inverse.

\begin{lem}(Admissible range of $\mathpzc{d}$ for~\eqref{eq::discrete-p-2N-zeroinput} in terms of maximum degree of the digraph)\label{cor::admissible_tau_consens-degree_dis}
Let $\GG$ be a SCWB digraph. 
Then, $\bar{\mathpzc{d}}$ in~\eqref{eq::stability con_dis} satisfies 
\begin{align}\label{eq::delay_bound_consens-out-deg-dis}
\bar{\mathpzc{d}}\geq\frac{1}{2}(\frac{1}{\beta\delta\dM}-1).
\end{align}
\end{lem}
\begin{proof}
Let $\Gamma$ be the  stability region introduced in \eqref{eq::stab_zone_disc} for a specific time delay, $\mathpzc{d}$. Invoking Gershgorin theorem, we know that all the eigenvalues of $\delta\vectsf{H}$ are located inside a circle which can be written in the polar form as 
$G=\{z=-r\ee^{i\,\theta}\in\complex|r=2\beta\delta\dM\,\cos(\theta)|-\frac{\pi}{2}\leq\theta\le\frac{\pi}{2}\}.$
Due to symmetricity of $G$ and $\Gamma$ , we just consider $0\leq\theta\leq\frac{\pi}{2}$ for simplification. $G$ lies inside $\Gamma$ if and only if $\phi=\frac{\pi}{2}-\theta,$ and 
   $|r|\leq |z|$.
Since $z=2\sin(\frac{\phi}{2\mathpzc{d}+1})\text{e}^{\ii(\frac{\pi}{2}+\phi)}$, it yields to $ \beta\,\delta\,\dM\sin(\phi)\leq\sin(\frac{\phi}{2\mathpzc{d}+1})$,
or $\beta\,\delta\,\dM\leq\,f(\phi)=\frac{\sin(\frac{\phi}{2\mathpzc{d}+1})}{\sin(\phi)}$.
Moreover, $f(\phi)$ is a strictly increasing function over $\phi\in[0,\frac{\pi}{2}]$, since
$$    \frac{\partial f(\phi)}{\partial \phi}=
   \frac{\cos(\phi)\cos(\frac{\phi}{2\mathpzc{d}+1})(\frac{1}{2\mathpzc{d}+1}\tan(\phi)-\tan(\frac{\phi}{2\mathpzc{d}+1}))}{\sin^2(\phi)}\geq0.
$$
Thus, the least value of $f(\phi)$ is obtained as $\phi\rightarrow0$ which is equal to $\frac{1}{2\mathpzc{d}+1}$. Thus, it implies that $\beta\delta\dM\leq\frac{1}{2\mathpzc{d}+1}$,
is a sufficient condition to guarantee the stability of the zero-input dynamics~\eqref{eq::discrete-p-2N-zeroinput}, which is equivalent to $\bar{\mathpzc{d}}\geq\frac{1}{2}(\frac{1}{\beta\delta\dM}-1).$  \end{proof}

\section{Numerical Simulations}\label{sec::num}
We analyze the robustness of algorithm~\eqref{eq::CT} and its discrete-time implementation~\eqref{eq::DCDisc} to communication delay for two academic examples taken from~\cite{SSK-JC-SM:15-ijrnc}. The network topology in these examples is given in
Figure~\ref{fig:network}~(a). We also use the network given in  Figure~\ref{fig:network}~(b) to study the effect of the maximum degree of a network on the admissible delay range. 
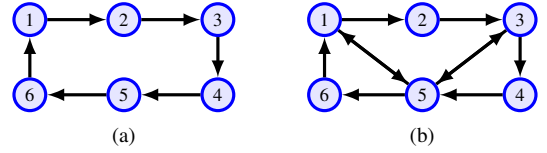
\begin{figure}
  \unitlength=0.6in \centering
  \subfloat[]{
    \centering
 \begin{tikzpicture}[auto,very thick,scale=1, every node/.style={scale=0.35} ]
                       \tikzstyle{every node}=[very thick, draw=blue, circle, minimum width=12pt, align=center]
     
     \node (1) [fill=blue!10] at (0-0.25,0){};
       \node (1p)[draw=none] at (0-0.25,0){{\scriptsize1}};
      \node[fill=blue!10] (2) at (1,0)  {};
       \node (2p)[draw=none] at (1,0)   {{\scriptsize 2}};
       \node (3) [fill=blue!10]at (2+0.25,0) {};
        \node (3p)[draw=none] at  (2+0.25,0) {{\scriptsize 3}};
    \node (4) [fill=blue!10] at (0-0.25,-1)  {};
     \node (4p)[draw=none] at  (0-0.25,-1)  {{\scriptsize 6}};
    \node (5)[fill=blue!10] at (1,-1){};) ;
\node (5p)[draw=none] at  (1,-1)  {{\scriptsize 5}};
    \node (6)[fill=blue!10] at (2+0.25,-1){};) ;
\node (6p)[draw=none] at  (2+0.25,-1)  {{\scriptsize 4}};
     \draw[-latex]  (1)->(2) ;
           \draw[-latex]  (2)->(3) ;
      \draw[-latex]  (4)->(1) ;
     \draw[-latex]  (5)->(4) ;
     \draw[-latex]  (6)->(5) ;
      \draw[-latex]  (3)->(6) ;
    \end{tikzpicture}  }
    \hspace{0.2in}
    \subfloat[]{
    \centering  
 \begin{tikzpicture}[auto,very thick,scale=1, every node/.style={scale=0.35} ]
                      \tikzstyle{every node}=[very thick, draw=blue, circle, minimum width=12pt, align=center]
     
    \node (1) [fill=blue!10] at (0-0.3,0){};
       \node (1p)[draw=none] at (0-0.3,0){{\scriptsize1}};
      \node[fill=blue!10] (2) at (1,0)  {};
       \node (2p)[draw=none] at (1,0)   {{\scriptsize 2}};
       \node (3) [fill=blue!10]at (2+0.3,0) {};
        \node (3p)[draw=none] at  (2+0.3,0) {{\scriptsize 3}};
    \node (4) [fill=blue!10] at (0-0.3,-1)  {};
     \node (4p)[draw=none] at  (0-0.3,-1)  {{\scriptsize 6}};
    \node (5)[fill=blue!10] at (1,-1){};) ;
\node (5p)[draw=none] at  (1,-1)  {{\scriptsize 5}};
    \node (6)[fill=blue!10] at (2+0.3,-1){};) ;
\node (6p)[draw=none] at  (2+0.3,-1)  {{\scriptsize 4}};
     \draw[-latex]  (1)->(2) ;
           \draw[-latex]  (2)->(3) ;
      \draw[-latex]  (4)->(1) ;
     \draw[-latex]  (5)->(4) ;
     \draw[-latex]  (6)->(5) ;
      \draw[-latex]  (3)->(6) ;
       \draw[-latex]  (3)->(5) ;
     \draw[-latex]  (1)->(5) ;
\draw[-latex]  (5)->(3) ;
\draw[-latex]  (5)->(1) ;
    \end{tikzpicture} }
        \caption{SCWB digraphs with edge weights $0$ and $1$. 
        }
    \label{fig:network}
\end{figure}

\emph{Continuous-time case:}
the  reference signals at each agent are
\begin{align*}
    &  r^1(t) = 0.55\,\sin(0.8\,t),~
 r^2(t) = 0.5\,\sin(0.7\,t)\!+\!0.5\,\sin(0.6\,t),\\
&r^3(t)=0.1\,t,\qquad\qquad\,
~r^4(t)=2\,\text{atan}(0.5\,t),\\
&r^5(t)=0.1\,\cos(2\,t),\quad\,
~r^6(t)=0.5\,\sin(0.5\,t). 
 \end{align*}
We set the parameter of the algorithm~\eqref{eq::CT} at $\beta=1$. The maximum admissible delay bounds  obtained using the result of Lemma~\ref{lem::admissible_tau_consens} for the topologies depicted in Figure~\ref{fig:network} are (a) $\bar{\tau}=0.52$ seconds and (b) $\bar{\tau}=0.41$ seconds. Note that as expected from~\eqref{eq::delay_bound_consens-out-deg}, for case (b) with $\dM=3$, the maximum admissible delay is less than case (a) with $\dM=1$.  Figure~\ref{fig::c} shows the time history of the tracking error of algorithm~\eqref{eq::CT} over the network topology of~Fig.~\ref{fig:network}~(a) for four different values of time delay  (a) $\tau=0$, (b) $\tau=0.2s$, (c) $\tau=0.4s$ and (d) $\tau=0.6s$. As this figure shows, by increasing  time delay, the rate of convergence of the algorithm decreases. Also,  as can be predicted from~\eqref{eq::Alg_D_ultimate_bound}, the tracking error increases. For case (d) the delay is beyond the admissible range, and as expected, results in instability. The convergence rate for each of the other aforementioned time delay is given by, respectively (a)
$\rho_{\tau}=0.5$  , (b) $\rho_{\tau}=0.28$ and (c) $\rho_{\tau}=0.11$.

\begin{figure}[t!]\label{Fig::con}
  \unitlength=0.5in
   \captionsetup{skip=13pt}\centering
  \subfloat{
    \includegraphics[trim={2pt 15pt 10pt 0},clip,scale=1]{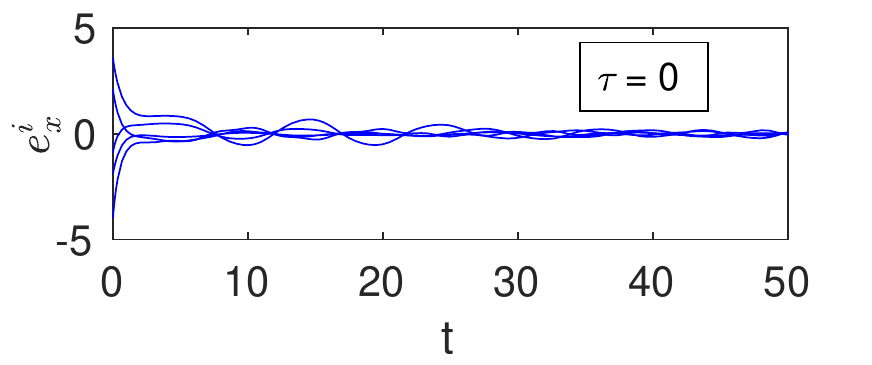} }
 \\\vspace{-0.15in}
 \subfloat{
    \includegraphics[trim={2pt 15pt 10pt 0},clip,scale=1]{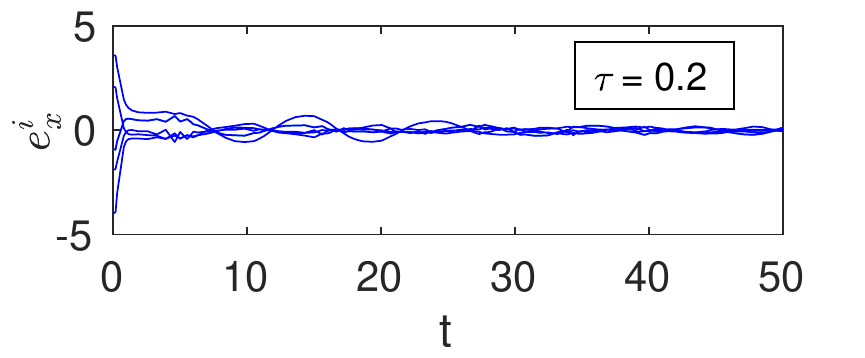} 
 }\\\vspace{-0.15in}
  \subfloat{
    \includegraphics[trim={2pt 15pt 10pt 0},clip,scale=1]{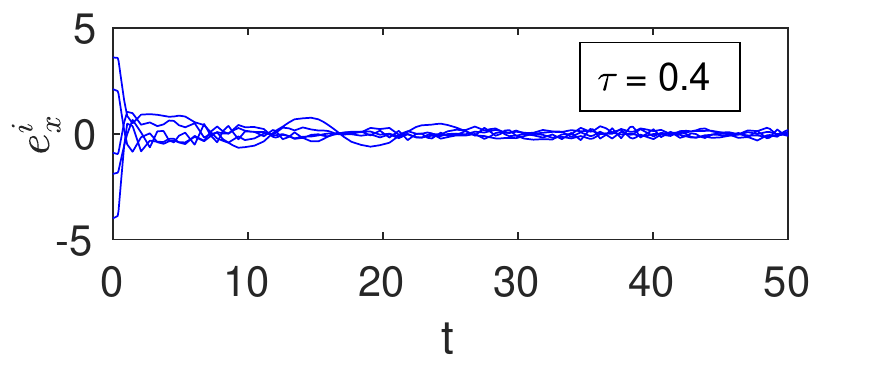} 
  }\vspace{-0.10in}
  \subfloat{
    \includegraphics[trim={2pt 0pt 10pt 0},clip,scale=1]{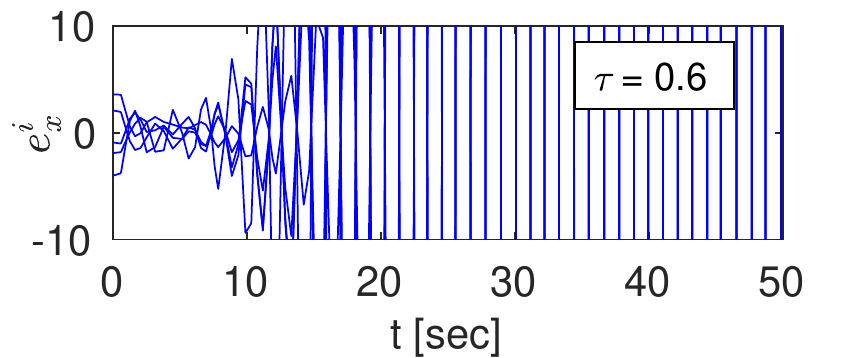} 
  }\vspace{-0.10in}
  \caption{
  Time history of the tracking error $e_x^i=x^i-\frac{1}{N}\sum_{j=1}^N\!\!\mathsf{r}^j$ of
  algorithm~\eqref{eq::CT} over the network of Figure~\ref{fig:network}~(a) in the absence and presence of communication~delay. 
  }\label{fig::c}
\end{figure}
\emph{Discrete-time case:} For discrete-time implementation, we use the same network depicted in Figure~\ref{fig:network}~(a) but with the reference input signal 
\begin{align*}
    \mathsf{r}^i(t(m))= 2+\sin\big({\omega(m)t(m)+\phi(m)}\big)+b^i \quad\quad m\in\mathbb{Z}_{\geq0}.
\end{align*}
Here, 
$\vect{b}^{\top}=[-0.55,1, 0.6,  -0.9, -0.6, 0.4]$, and $\omega$ and $\phi$ are random signals with Gaussian distributions, $N(0,0.25)$ and $N(0,(\frac{\pi}{2})^2)$, respectively. ~\cite{SSK-JC-SM:15-ijrnc} states that these reference signals correspond to a group of sensors with bias $b^i$ which sample a process at sampling times $t(m)$, $m\in\mathbb{Z}_{\geq0}$. Here, we assume that the data is sampled at $\Delta\,t(m)=5s$, for all $m\in\mathbb{Z}_{\geq0}$. Each sensor needs to obtain the average of the measurements across the network before the next sampling time.  To obtain the average we implement~\eqref{eq::DCDisc} with sampling stepsize $\delta$ and input  $\mathsf{r}^i(k)=\mathsf{r}^i(k\delta)$, $i\in\VV$, which between sampling times $t(m)$ and $t(m+1)$ is fixed at $\mathsf{r}^i(t(m))$.  
Figure~\ref{fig::d} demonstrates the result of simulation for $\beta=1$ and $\delta=0.19\,s$. The admissible delay for this case is $\bar{\mathpzc{d}}=2$. Figure~\ref{fig::d} shows the time history of tracking error  for different amounts of delay (a) $\mathpzc{d}=0$, (b) $\mathpzc{d}=1$, (c) $\mathpzc{d}=2$, and (d) $\mathpzc{d}=3$.  As shown, the steady error goes up as the delay increases. Moreover, the rate of convergence for each case is obtained by solving optimization problem~\eqref{eq::LMI}. The optimal upper bound specifications for each case correspond to (a)
$\bar{\kappaa}_{\,\,\mathpzc{d}}=7.2$, $\bar{\omega}_\mathpzc{d}=0.16$ (b) $\bar{\kappaa}_{\,\,\mathpzc{d}}=68$, $\bar{\omega}_\mathpzc{d}=0.96$ and (c) $\bar{\kappaa}_{\,\,\mathpzc{d}}=70$, $\bar{\omega}_\mathpzc{d}=0.99$, respectively. In case (d), the delay is outside admissible range and as expected the algorithm is unstable and diverges.

\begin{figure}[t!]
  \unitlength=0.5in
   \captionsetup{skip=15pt}\centering
  \subfloat{
    \includegraphics[trim={2pt 15pt 10pt 0},clip,scale=1]{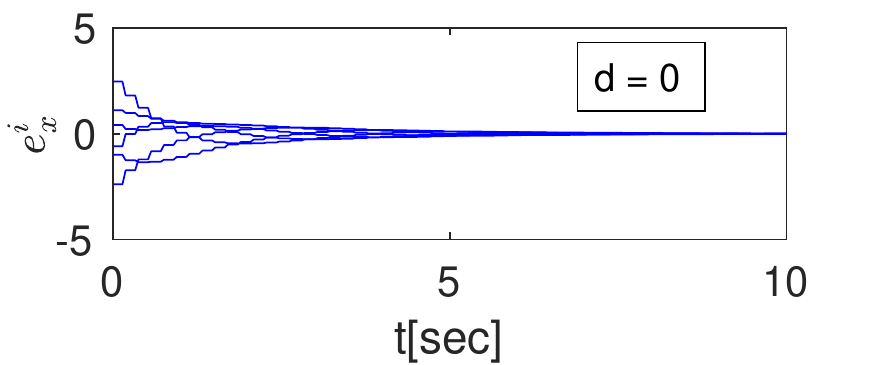}
  }\\\vspace{-0.15in}
 \subfloat{
    \includegraphics[trim={2pt 15pt 10pt 0},clip,scale=1]{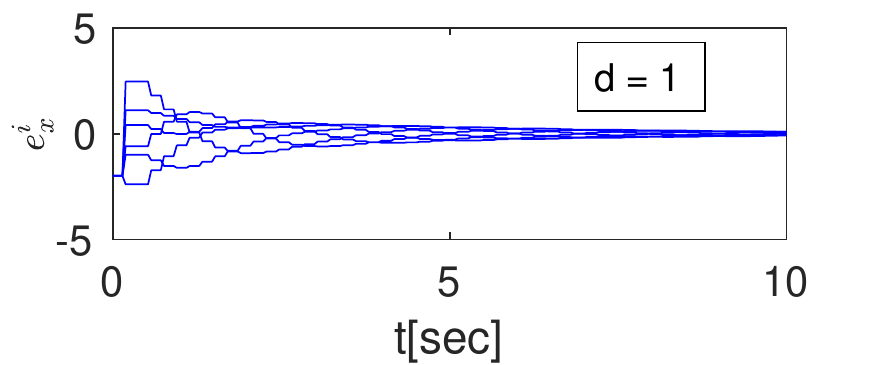} 
 }\\\vspace{-0.15in}
  \subfloat{
    \includegraphics[trim={2pt 15pt 10pt 0},clip,scale=1]{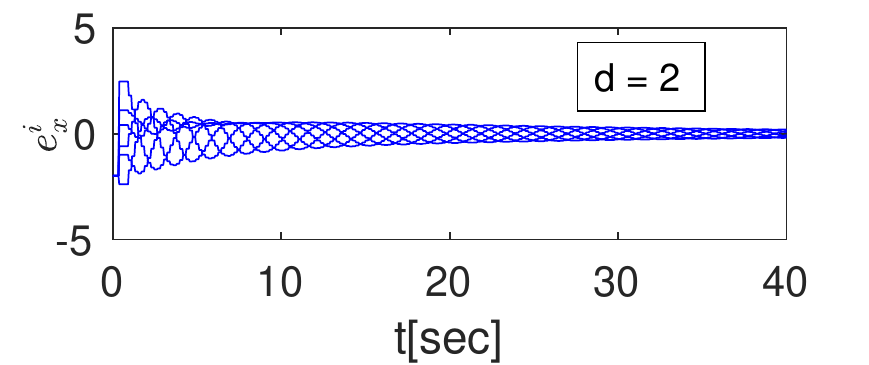} 
  }\vspace{-0.10in}
  \subfloat{
    \includegraphics[trim={2pt 0pt 10pt 0},clip,scale=1]{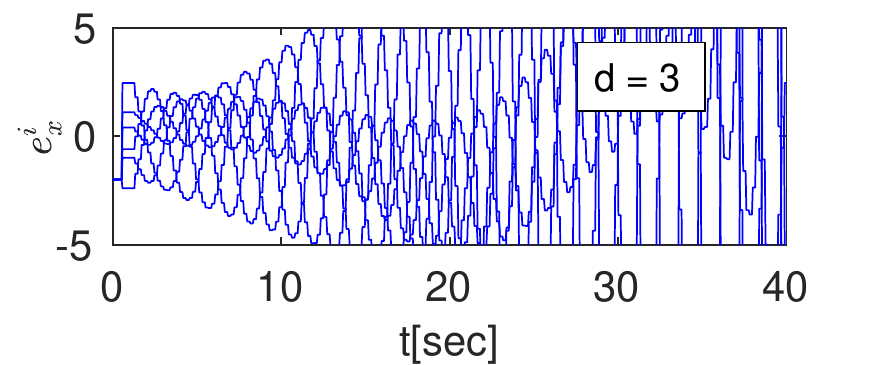} 
  }\vspace{-0.10in}
  \caption{
  Time history of the tracking error $e_x^i=x^i-\frac{1}{N}\sum_{j=1}^N\!\!\mathsf{r}^j$ of algorithm~\eqref{eq::DCDisc} over the network of Figure~\ref{fig:network}~(a) in the absence and presence of communication~delay. The vertical dashed lines show the data sampling times (every 5 seconds).}
  \label{fig::d}
\end{figure}

\section{Conclusions}\label{sec::conclu}
We studied the robustness of a  dynamic average consensus algorithm to  fixed communication delays. Our study included both the continuous-time and the discrete-time implementations of this algorithm. For both implementations over strongly connected and weight-balanced  digraphs, we (a) characterized the admissible delay range, (b) studied the effect of communication delay on the rate of convergence and the tracking error, (c) obtained  upper bounds for them based on the value of the communication delay and the network's and the algorithm's parameters and (d) showed that the size of the admissible delay range has an inverse relation with the maximum degree of the interaction topology. Future work will be devoted to investigating the effect of uncommon and time-varying communication delay on the algorithms' stability and convergence and expanding also our results to switching~graphs. 

\bibliographystyle{ieeetr}%
\bibliography{bib/alias,bib/Reference} 

\section*{Appendix: stability of time-delay systems of retarded type}\label{sec::DDE-systems}

\renewcommand{\theequation}{A.\arabic{equation}}
\renewcommand{\thethm}{A.\arabic{thm}}
\renewcommand{\thelem}{A.\arabic{lem}}
\renewcommand{\thedefn}{A.\arabic{defn}}
The exponential stability of a continuous-time linear DDE  
\begin{align}\label{eq::delay-sys}
&\dvect{x}(t)=\vect{A}\vect{x}(t-\tau),\quad  t\in\real_{\geq0},\\
&\vect{x}(0)=\vect{x}_{0}\in\real^n,\quad \vect{x}(t) =\vect{g}(t)~\text{for~}  t\in [-\tau,0),\nonumber
\end{align}
where $\vect{g}(t)$  is a continuously differential pre-shape function and $\tau\in\realpositive$ is  the time-delay,  is defined as follows. 
\begin{defn}[Internal stability of~\eqref{eq::delay-sys}~\cite{SN:01}]\label{def::expen-stabil-continuous-delay}
For a given $\tau\in\realpositive$ the trivial solution $\vect{x}\equiv \vect{0}$ of the zero input system of~\eqref{eq::delay-sys}  is said to be exponentially stable iff there exists a $\kappaa_{\,\,\tau}\in\realpositive$ and an $\rho_{\tau}\in\realpositive$ such that for the given initial conditions the solution satisfies the inequality below
\begin{align}\label{eq::expon_stabili}\|\vect{x}(t)\|\leq \kappaa_{\,\,\tau}\text{e}^{-\rho_{\tau}\,t}\underset{\eta\in[-\tau,0]}{\text{sup}}\|\vect{x}(\eta)\|, \quad t\in\realnonnegative.\end{align}   \end{defn}
The exponentially stability for  discrete-time time-delay system
\begin{align}\label{eq::delay-sys-discrete}
\dvect{x}(k+1)&=\vect{x}(k)+\vect{A}\vect{x}(k-\mathpzc{d}),\quad k\in\mathbb{Z}_{k\geq 0},\\
\vect{x}(0) &=\vect{x}_{0}\real^n,  ~~\vect{x}(k) =\vect{g}(k),   \quad  k\in\mathbb{Z}_{-\mathpzc{d}}^{-1},\nonumber
\end{align}
where $\mathpzc{d}\in\mathbb{Z}_{k\geq 0}$ is the fixed time-step delay and $\vect{g}(k)$ is a pre-shape function, is defined as follows
\begin{defn}[Internal stability of~\eqref{eq::delay-sys-discrete}]
For a given $\mathpzc{d}\in\mathbb{Z}_{k\geq 0}$, the trivial solution $\vect{x}\equiv \vect{0}$ of the zero input system of~\eqref{eq::delay-sys-discrete}  is said to be exponentially stable iff there exists a $\kappaa_{\,\,\mathpzc{d}}\in\realpositive$ and an $0<\omega_{\mathpzc{d}}<1$  such that for the given initial conditions the solution satisfies the inequality below
\begin{align}\label{eq::expon_stabili-discrete}\|\vect{x}(k)\|\leq \kappaa_{\,\,\mathpzc{d}}\,\omega_{\mathpzc{d}}^k\,\underset{k\in \mathbb{Z}_{-\mathpzc{d}}^0}{\text{sup}}\|\vect{x}(k)\|, \quad k\in\mathbb{Z}_{\geq0}.\end{align}  \end{defn}

Next, we develop an auxiliary result which we
use in proof of our main result given in Theorem~\ref{thm::Alg_delay}.
\begin{lem}[Exponential upper bound on a discrete-time delay dynamics]\label{lem::exponential_rate_bound_dc}
Consider the discrete-time time-delay system 
\begin{align}\label{eq::discrete-system-zero-input}
\vect{x}(k+1)&=\vect{x}(k)+\vect{A}\vect{x}(k-\mathpzc{d}),~~~k\in\mathbb{Z}_{\geq 0},\end{align} 
where $\mathpzc{d}\in\mathbb{Z}_{\geq0}$, $\vect{A}\in\real^{n\times n}$ and the initial conditions are
$\vect{x}(k)\in\real^{n}$ for $k\in\{-\mathpzc{d},\cdots,0\}$. Assume that the admissible delay range of~\eqref{eq::discrete-system-zero-input} is non-empty, i.e., there exists a $\bar{\mathpzc{d}}\in\mathbb{Z}_{>0}$ such that for $\mathpzc{d}\in[0,\bar{\mathpzc{d}}]$ the system~\eqref{eq::discrete-system-zero-input} is exponentially stable. 
Then, for every  $\mathpzc{d}\in[0,\bar{\mathpzc{d}}]$ there always exists a positive definite $\vect{Q}\in\real^{(\mathpzc{d}+1)n\times(\mathpzc{d}+1)n}$ and ${\kappaa}_{\,\,\mathpzc{d}},\omega_{\mathpzc{d}}\in\real_{>0}$ that satisfy
 \begin{subequations}\label{eq::expon-condi-delay}
   \begin{align}
&\frac{1}{{\kappaa}_{\,\,\mathpzc{d}}} \,\vect{I}\leq \vect{Q}\leq \vect{I},\quad 0<\omega_\mathpzc{d}^2<1,\quad \kappaa_{\,\,\mathpzc{d}}>1,\\
&\vect{A}_{\text{aug}}^\top\vect{Q}\vect{A}_{\text{aug}}-\vect{Q}\leq -(1-\omega_{\mathpzc{d}}^2)\,\vect{I}.
\end{align}
\end{subequations}

 Here, $\vect{A}_{\text{aug}}=\begin{bmatrix}\vect{0}_{\mathpzc{d}n\times n}&\vect{I}_{\mathpzc{d}n}\\
\vect{A}&\begin{bmatrix}\vect{0}_{n\times (\mathpzc{d}-1)n}&\vect{I}_n\end{bmatrix}\end{bmatrix}$. Moreover,
\begin{itemize}
\item[(a)] $\|\vect{x}(k)\|\leq \sqrt{\mathpzc{d}+1}\,{\kappaa}_{\,\,\mathpzc{d}}\,\omega_{\mathpzc{d}}^k\,\underset{k\in \mathbb{Z}_{-\mathpzc{d}}^0}{\text{sup}}\|\vect{x}(k)\|, \quad k\in\mathbb{Z}_{\geq0}$.
\item[(b)] when  $\vect{x}(k)=\vect{0}$ for $k\in\mathbb{Z}_{-\mathpzc{d}}^{-1}$, we have
\begin{align}\label{eq::expon_stabili-discrete-zero-init}\|\vect{x}(k)\|\leq {\kappaa}_{\,\,\mathpzc{d}}\,\omega_{\mathpzc{d}}^k\,\|\vect{x}(0)\|, \quad k\in\mathbb{Z}_{\geq0}.
\end{align}
\end{itemize}
\end{lem}
\begin{proof}
For a $\mathpzc{d}$ in the admissible delay range of~\eqref{eq::discrete-system-zero-input}, let $\vect{x}_{\text{aug}}(k)=\begin{bmatrix}\vect{x}(k-\mathpzc{d})^\top&\cdots&\vect{x}(k-1)^\top&\vect{x}(k)^\top\end{bmatrix}^\top$. Then,
\begin{align}\label{eq::disc-aug}
 \vect{x}_{\text{aug}}(k+1)&=\!
\vect{A}_{\text{aug}}\vect{x}_{\text{aug}}(k),~ k\in\mathbb{Z}_{\geq 0},~\vect{x}_{\text{aug}}\in\real^{(\mathpzc{d}+1)n},
\end{align}
with $\vect{A}_{\text{aug}}$ as defined in the statement. 
Because~\eqref{eq::discrete-system-zero-input} is exponentially stable, the augmented state equation~\eqref{eq::disc-aug} is also exponentially stable. Then, by virtue of
results on exponential stability of LTI discrete-time systems~\cite[Theorem 23.3]{WJR:93}, we have the guarantees that there exists a symmetric positive definite  $\vect{P}\in\real^{n\times n}$ and scalars $\eta,\rho,\eps\in\real_{>0}$ that satisfy
\begin{align}\label{eq::LMI-exp-bound-discrete}
&\eta \,\vect{I}\leq \vect{P}\leq \rho\,\vect{I},\quad\vect{A}_{\text{aug}}^\top\vect{P}\vect{A}_{\text{aug}}-\vect{P}\leq -\eps\,\vect{I}.
\end{align}
Moreover, for $\kappaa_{\,\,\mathpzc{d}}=\rho/\eta$ and $0<\omega_{\,\,\mathpzc{d}}^2=1-\eps/\rho<1$, we have
\begin{equation}\label{eq::x_aug_bound}
\|\vect{x}_{\text{aug}}(k)\|\leq \kappaa_{\,\,\mathpzc{d}} \,\omega_{\,\,\mathpzc{d}}^k\|\vect{x}_{\text{aug}}(0)\|,\quad k\in\mathbb{Z}_{\geq 0}.
\end{equation}
By applying the change of variables $\vect{Q}=\frac{1}{\rho}\vect{P}$, ${\kappaa}_{\,\,\mathpzc{d}}=\rho/\eta>0$ and $0<\omega_{\mathpzc{d}}^2=1-\eps/\rho<1$, we can represent~\eqref{eq::LMI-exp-bound-discrete} in its equivalent form in~\eqref{eq::expon-condi-delay}. Moreover, because $\|\vect{x}(k)\|\leq 
\|\vect{x}_{\text{aug}}(k)\|$ and $\|\vect{x}_{\text{aug}}(0)\|=\sqrt{
\|\sum_{j=0}^{-d}\vect{x}(j)\|^2}
\leq \sqrt{\mathpzc{d}+1} \underset{k\in \mathbb{Z}_{-\mathpzc{d}}^0}{\text{sup}}\|\vect{x}(k)\|$ and $\|\vect{x}(k)\|\leq 
\|\vect{x}_{\text{aug}}(k)\|$, we can use~\eqref{eq::x_aug_bound} to confirm claim (a) in the statement. On the other hand, because $\vect{x}(k)=\vect{0}$ for $k\in\mathbb{Z}_{-\mathpzc{d}}^{-1}$,~\eqref{eq::x_aug_bound} also guarantees claim (b) in the statement.
\end{proof}

\end{document}